\documentclass[a4paper]{llncs}

\usepackage{amsmath,amssymb}

\usepackage{paralist}
\usepackage{microtype}

\usepackage{algorithm,algorithmic}

\newcommand{\INLIF}[2]{\STATE\algorithmicif\ #1\ \algorithmicthen\ #2\ \algorithmicendif}
\newcommand{\INLIFELSE}[3]{\STATE\algorithmicif\ #1\ \algorithmicthen\ #2\ \algorithmicelse\ #3\ \algorithmicendif}
\newcommand{\RET}{\algorithmicreturn\ }
\newcommand{\algorithmicmatch}{\textbf{match}}
\newcommand{\algorithmicendmatch}{\textbf{end match}}
\newcommand{\MATCH}[1]{\STATE\algorithmicmatch\ #1\begin{ALC@g}}
\newcommand{\ENDMATCH}{\end{ALC@g}\STATE\algorithmicendmatch}
\newcommand{\algorithmiccase}{\textbf{case}}
\newcommand{\algorithmicendcase}{}
\newcommand{\CASE}[1]{\STATE\algorithmiccase\ #1:\begin{ALC@g}}
\newcommand{\ENDCASE}{\end{ALC@g}\algorithmicendcase}

\setboolean{ALC@noend}{true}
\renewcommand{\INLIF}[2]{\STATE\algorithmicif\ #1\ \textbf{then}\ #2}
\renewcommand{\INLIFELSE}[3]{\STATE\algorithmicif\ #1\ \textbf{then}\ #2\ \textbf{else}\ #3}

\usepackage{subfigure}
\usepackage{tikz}
\usetikzlibrary{calc,arrows,positioning}
\tikzstyle{state}=[circle,draw,inner sep=2pt,node distance=50pt]
\tikzstyle{marked}=[state,double]
\tikzstyle{trans}=[-angle 90]
\tikzstyle{acc}=[node distance=15pt,align=center,font=\scriptsize]

\usepackage[hidelinks]{hyperref}

\DeclareMathOperator{\dom}{dom}
\DeclareMathOperator{\ready}{ready}
\DeclareMathOperator{\Acc}{Acc}
\DeclareMathOperator{\pre}{pre^{*}}
\DeclareMathOperator{\prep}{pre^{+}}
\DeclareMathOperator{\post}{post^{*}}
\DeclareMathOperator{\postp}{post^{+}}
\DeclareMathOperator{\Cycle}{Loop}
\DeclareMathOperator{\Dead}{DeadFree}
\DeclareMathOperator{\Live}{LiveFree}
\DeclareMathOperator{\Un}{Un}
\DeclareMathOperator{\Compat}{Compat}

\newcommand{\mass}{MAS}
\newcommand{\masss}{MAS}
\newcommand{\massp}{MASp}

\newcommand{\prequotient}{/\!\!/}
\newcommand{\compreach}{\ensuremath{\sim_\mathcal{T}}}
\newcommand{\unncomp}{\ensuremath{\sim_\mathcal{U}}}
\newcommand{\CycleI}{\Cycle_{\models}}

\title{Quotient of Acceptance Specifications under Reachability
Constraints}
\author{Guillaume Verdier \and Jean-Baptiste Raclet}
\institute{IRIT/CNRS, 118 Route de Narbonne, F-31062 Toulouse Cedex 9, France\\
\email{$\{$verdier,raclet$\}$@irit.fr}}

\begin{document}

\maketitle

\begin{abstract}
  The quotient operation, which is dual to the composition, is crucial
  in specification theories as it allows the synthesis of missing
  specifications and thus enables incremental design. In this paper, we
  consider a specification theory based on marked acceptance
  specifications (MAS) which are automata enriched with variability
  information encoded by acceptance sets and with reachability
  constraints on states. We define a sound and complete quotient for
  MAS hence ensuring reachability properties by construction.
\end{abstract}

\section{Introduction}

Component-based design aims at building complex reactive systems by
assembling components, possibly taken off-the-shelf. This approach can
be supported by a specification theory in which requirements
correspond to specifications while components are models of the
specifications. Such theories come equipped with a set of operations
enabling modular system design.

Several recent specification theories are based on modal
specifications~\cite{Raclet2011,LuttgenV13,ChenCJK12}
including in timed~\cite{Bertrand2011,David2010} or
quantitative~\cite{BauerFJLLT13} contexts and with
data~\cite{BauerLLNW14}. In this paper, we introduce \emph{marked
  acceptance specifications} (MAS): they are based on an extension of
modal specifications, called acceptance specifications, which we
enrich with marked states to model reachability objectives. This last
addition is needed to model session terminations, component
checkpoints or rollbacks.

A crucial feature in a specification theory is the operation of
quotient. Let~$S_1$ be the specification of a target system and $S_2$
be the specification of an available black-box component. The
specification $S_1 / S_2$ characterizes all the components that, when
composed with any model of $S_2$, conform with $S_1$. In other words,
$S_1 / S_2$ tells what remains to be implemented to realize $S_1$
while reusing a component doing $S_2$. By allowing to characterize
missing specifications, quotient thus enables incremental design
and component reuse.

The quotient of specifications also plays a central role in
contract-based design. In essence, a contract describes what a system
should guarantee under some assumptions about its context of use. It
can be modeled as a pair of specifications $(A, G)$ for, respectively,
the assumptions and the guarantees. Satisfiability of a contract then
corresponds to the satisfiability of the specification $G / A$
(see~\cite{Benveniste2011} for more explanations on contract
satisfaction). 

\emph{Our contribution.} Firstly, we define MAS and their
semantics. The included marked states allow to specify reachability
objectives that must be fulfilled by any model of the MAS. A MAS then
characterizes a set of automata called \emph{terminating} as
they satisfy the reachability property telling that a marked state
can always be reached.

Secondly, we study the compositionality of MAS. We define a
compatibility criterion such that two MAS $S_1$ and $S_2$ are
compatible if and only if the product of any models of $S_1$ and $S_2$
is terminating.
Further, given two incompatible MAS $S_1$ and $S_2$, we propose a
construction to refine $S_1$ into the most general $S'_1$ such that
$S'_1$ and $S_2$ become compatible.

Last, we define the quotient of MAS. This is a two-step construction
that makes use of the previous cleaning construction.  The operation
is shown to be sound and complete. 

\emph{Related work.} 
Modal specifications~\cite{KT88} enriched with marked states (MMS)
have been introduced in~\cite{Darondeau2010} for the supervisory
control of services. Product of MMS has been investigated
in~\cite{CaillaudR12}. These papers did not show the need for the more
expressive framework of MAS as quotient was not considered.
Acceptance specifications have first been proposed
in~\cite{Raclet2008} based on~\cite{Hennessy85}. Their
non-deterministic version is named Boolean MS
in~\cite{BenesKLMS11}. The LTL model checking of MS has been studied
in~\cite{BenesCK11}. However, the reachability considered in this
paper can be stated in CTL by \verb+AG(EF(final))+ and cannot be
captured in LTL.

Quotient of modal and acceptance specifications has been studied
in~\cite{Raclet2008,ChenCJK12} and in~\cite{BenesDFKL13}
for the non-deterministic case. It has also been defined for
timed~\cite{Bertrand2011,David2010} and
quantitative~\cite{BauerFJLLT13} extensions of modal
specifications. None of these works consider reachability constraints.

\emph{Outline of the paper.} We recall some definitions about automata
and introduce \masss{} in \autoref{sec:definitions}. Then, we define
the pre-quotient operation in \autoref{sec:prequo} which only
partially solves the problem as it does not ensure the reachability of
marked states. In \autoref{sec:compreach}, we give a criterion of
compatible reachability telling whether the product of the models of
two MAS is always terminating. When this condition of compatible
reachability is not met, it is possible to impose some constraints on
one of the specifications in order to obtain it, as shown in
\autoref{sec:clean}. Based on this construction,
\autoref{sec:quotient} finally defines the quotient operation on MAS.
\section{Modeling with Marked Acceptance Specifications}
\label{sec:definitions}

\subsection{Background on automata}

A (deterministic) \emph{automaton} over an alphabet $\Sigma$ is a
tuple $M = (R, r^0, \lambda, G)$ where $R$ is a finite set of states,
$r^0 \in R$ is the initial state, $\lambda : R \times \Sigma
\rightharpoonup R$ is the labeled transition map and $G \subseteq R$
is the set of marked states. The set of \emph{fireable} actions from a
state $r$, denoted $\ready (r)$, is the set of actions $a$ such that
$\lambda (r, a)$ is defined.

Given a state $r$, we define $\pre(r)$ and $\post(r)$ as the smallest sets such
that $r \in \pre(r)$, $r \in \post(r)$ and for any $r'$, $a$ and $r''$ such that
$\lambda(r', a) = r''$, $r' \in \pre(r)$ if $r'' \in \pre(r)$ and $r'' \in
\post(r)$ if $r' \in \post(r)$. We also define $\prep(r)$ as the union of
$\pre(r')$ for all $r'$ such that $\exists a: \lambda(r', a) = r$ and
$\postp(r)$ as the union of $\post(\lambda(r, a))$ for all $a \in \ready(r)$.
Let $\Cycle(r) = \prep(r) \cap \postp(r)$.

Two automata $M_1$ and $M_2$ are bisimilar iff there exists a
simulation relation $\pi: R_1 \times R_2$ such that $(r_1^0, r_2^0)
\in \pi$ and for all $(r_1, r_2) \in \pi$, $\ready (r_1) = \ready
(r_2) =\nobreak Z$, $r_1 \in G_1$ iff $r_2 \in G_2$ and for any $a \in Z$,
$(\lambda (r_1, a), \lambda (r_2, a)) \in \pi$.

The \emph{product} of two automata $M_1$ and $M_2$, denoted $M_1
\times M_2$, is the automaton $(R_1 \times R_2, (r_1^0, r_2^0),
\lambda, G_1 \times G_2)$ where $\lambda ((r_1, r_2), a)$ is
defined as the pair $(\lambda_1 (r_1, a), \lambda_2 (r_2, a))$ when
both $\lambda_1 (r_1, a)$ and $\lambda_2 (r_2, a)$ are defined.

Given an automaton $M$ and a state $r$ of $M$, $r$ is a
\emph{deadlock} if $r \not\in G$ and $\ready(r) = \emptyset$; $r$
belongs to a \emph{livelock} if $\Cycle(r) \neq \emptyset$, $G \cap
\Cycle(r) = \emptyset$ and there is no transition $\lambda(r', a) =
r''$ such that $r' \in \Cycle(r)$ and $r'' \not\in \Cycle(r)$. An
automaton is \emph{terminating} if it is deadlock-free and
livelock-free.

\subsection{Marked Acceptance Specification}

We now enrich acceptance specifications~\cite{Raclet2008} with marked
states to model reachability constraints. The resulting formalism
allows to specify a (possibly infinite) set of terminating automata
called models.

\begin{definition}[MAS]
  A \emph{marked acceptance specification} (MAS) over an alphabet
  $\Sigma$ is a tuple $S = (Q, q^0, \delta, \Acc, F)$ where $Q$ is a
  finite set of states, $q^0 \in Q$ is the initial state,
  $\delta : Q \times \Sigma \rightharpoonup Q$ is the labeled
  transition map, $\Acc: Q \rightarrow 2^{2^\Sigma}$ associates to
  each state its acceptance set and $F \subseteq Q$ is a set of marked
  states. 
\end{definition}

Basically, an acceptance set is a set of sets of actions a model of
the specification is ready to engage in. The \emph{underlying}
automaton associated to $S$ is $\Un(S) = (Q, q^0, \delta, F)$. We only
consider MAS such that $\Un(S)$ is deterministic.

\begin{definition}[Satisfaction]\label{def:sat}
  A terminating automaton $M$ \emph{satisfies} a \mass{} $S$, denoted
  $M \models S$, iff there exists a simulation relation $\pi \subseteq
  R \times Q$ such that $(r^0, q^0) \in \pi$ and for all $(r, q) \in
  \pi$: $\ready (r) \in \Acc (q)$; if $r \in G$ then $q \in F$; and, for
  any $a$ and $r'$ such that $\lambda (r, a) = r'$, $(r', \delta (q,
  a)) \in \pi$. $M$ is called a \emph{model} of~$S$.
\end{definition}

\begin{example}
  A MAS is depicted in \autoref{fig:prequo:S1}. Marked states are
  double-circled while the acceptance sets are indicated near their
  associated state. The terminating automata $M'$ and $M''$ in
  \autoref{fig:mod:M11} and \autoref{fig:mod:M12} are models of $S_1$
  because of the respective simulation relation $\pi' = \{ (0',0),
  (1',1) \}$ and $\pi'' = \{ (0'',0), (1'',0), (2'',1) \}$.
  Observe that the transitions labeled by $b$ and $c$ are optional in
  state $0$ from the MAS $S_1$ as these actions are not present in all
  sets in $\Acc(0)$ and thus may not be present in any model of the
  specification. Moreover, state $1$ in $S_1$ is marked to encode the
  constraint that it must be simulated in any model. As a result,
  although the actions $b$ and $c$ are optional, at least one of the
  two must be present in any model of $S_1$. This kind of constraint
  entails that MAS are more expressive than MS.
\end{example}

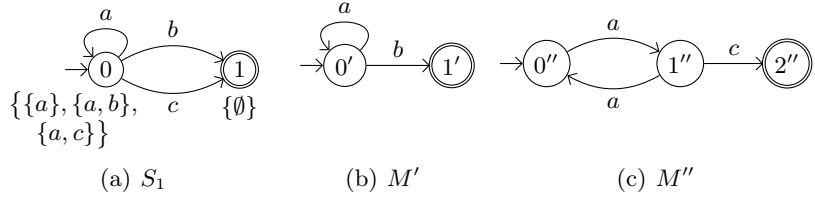
\begin{figure}
\centering
\subfigure[$S_1$]{
  \begin{tikzpicture}
    \node[state] (q0) {$0$};
    \node[below left=0pt and -20pt of q0,align=center] {$\bigl\{\{a\},\{a,b\},$\\$\{a,c\}\bigr\}$};
    \node[marked,right of=q0] (q2) {$1$};
    \node[below=0pt of q2] {$\{\emptyset\}$};
    \draw[trans] ($ (q0.west) + (-.3, 0) $) to (q0);
    \draw[trans,min distance=20pt] (q0) to node[above] {$a$} (q0);
    \draw[trans,bend left] (q0) to node[above] {$b$} (q2);
    \draw[trans,bend right] (q0) to node[below] {$c$} (q2);
  \end{tikzpicture}
  \label{fig:prequo:S1}
}
\subfigure[$M'$]{
  \begin{tikzpicture}
    \node[state] (q0) {$0'$};
    \node[marked,right of=q0,node distance=40pt] (q2) {$1'$};
    \draw[trans] ($ (q0.west) + (-.3, 0) $) to (q0);
    \draw[trans,min distance=20pt] (q0) to node[above] {$a$} (q0);
    \draw[trans] (q0) to node[above] {$b$} (q2);
    \node[below=20pt of q0] {$\ $};
  \end{tikzpicture}
  \label{fig:mod:M11}
}
\subfigure[$M''$]{
  \begin{tikzpicture}
    \node[state] (q0) {$0''$};
    \node[state,right of=q0] (q1) {$1''$};
    \node[marked,right of=q1,node distance=40pt] (q4) {$2''$};
    \draw[trans] ($ (q0.west) + (-.3, 0) $) to (q0);
    \draw[trans,bend left] (q0) to node[above] {$a$} (q1);
    \draw[trans,bend left] (q1) to node[below] {$a$} (q0);
    \draw[trans] (q1) to node[above] {$c$} (q4);
    \node[below=20pt of q0] {$\ $};
  \end{tikzpicture}
  \label{fig:mod:M12}
}
\caption{Example of MAS with two models}
\label{fig:mas}
\end{figure}

The introduced semantic induces some simplifications in
the structure of the \mass{} that we discuss now. This will then lead
to the definition of an associated normal form. \\

\begin{compactitem}
\item \emph{attractability.} A \mass{} is said attracted in $q$ when
  $\post(q) \cap F \neq \emptyset$.

\item \emph{$\Acc$-consistency.} A state $q$ is $\Acc$-consistent when
  $\Acc(q) \neq \emptyset$.

\item \emph{$F,\Acc$-consistency.} A state $q$ is $F,\Acc$-consistent
  when $\emptyset \in \Acc(q)$ implies $q \in F$.

\item \emph{$\delta,\Acc$-consistency.} A state $q$ is
  $\delta,\Acc$-consistent when, for any action $a\in\Sigma$,
  $\delta(q,a)$ is defined if and only if there exists an $X \in
  \Acc(q)$ such that $a \in X$.
\end{compactitem}

\begin{remark}\label{rk:sem}
  When $\Acc(q) = \emptyset$, $q$ cannot belong to a simulation
  relation stating that $M \models S$ as we cannot find an $X \in
  \Acc(q)$ such that $\ready(r) = X$ for some $r$. Moreover, when
  $\emptyset \in \Acc(q)$ and $(r,q) \in \pi$, we can have $\ready(r)
  = \emptyset$ that is, there is no outgoing transition from $r$. As
  $M$ is terminating, this requires that $r$ is marked and thus, $q$
  is also marked.
\end{remark}

 \begin{definition}[Normal form]\label{def-nf}
   A \mass{} is in \emph{normal form} if it is attracted,
   $\Acc$-consistent, $F,\Acc$-consistent and $\delta,\Acc$-consistent in
   every state $q$.
\end{definition}

\begin{theorem}\label{prop:nf}
  Every marked acceptance specification is equivalent to a marked
  acceptance specification in normal form.
\end{theorem}

First, we introduce a specification $S_\bot$ that admits no model. We
assume that $S_\bot$ is in normal form. Now proof of the previous
theorem is by construction of a MAS in normal form $\rho({S})$ and
then proving that $S$ and $\rho({S})$ are equivalent. This
construction is detailed in \autoref{algo:fn} which returns a MAS
corresponding to $\rho({S})$; it defines a pruning operation which
removes all the states which are not attracting or not
$\Acc$-consistent and updates the acceptance sets and the transition
map to enforce $F,\Acc$-consistency and $\delta,\Acc$-consistency.

\begin{proof}

  Following \autoref{def-nf}, four cases may lead to detect that a
  \mass{} $S$ is not in normal form. For each of them, we associate a
  construction rule to obtain $\rho({S})$ that we analyze now:

\begin{compactitem}
\item if $\neg \Acc$-consistent($q$) then $q$ cannot belong to a
  simulation relation stating that $M \models S$ as explained in
  \autoref{rk:sem}. We thus remove the state from $S$ (lines
  \ref{alg:c1:start} to \ref{alg:c1:end} of \autoref{algo:fn});
\item if $\neg$ attracted($q$) then no marked state is reachable from
  $q$ in $S$. We thus cannot include $q$ in a simulation relation to
  build a terminating model. As in the previous case, we remove $q$
  from $S$ (lines \ref{alg:c1:start} to \ref{alg:c1:end} of
  \autoref{algo:fn});
\item if $\neg F,\Acc$-consistent($q$), we remove $\emptyset$ from
  $\Acc(q)$ (line \ref{alg:c2} of \autoref{algo:fn}). This is a direct
  consequence of the fact that $\emptyset \in \Acc(q)$ is only
  relevant when $q$ is marked as advocated in \autoref{rk:sem};
\item last, the combination of items 1 and 3 of \autoref{def:sat}
  indicates that a ready set $X$ is relevant in $q$ if and only if the
  transition map is defined from $q$ for any $a \in X$. In lines
  \ref{alg:c3:start} to \ref{alg:c3:end} of \autoref{algo:fn}, we then
  update the acceptance sets and the transition map to make them
  consistent. 
\end{compactitem}

As a result, none of these four cases affects the set of models of
the received \mass{}. These construction rules are iteratively applied
until a fix-point is reached. The algorithm always finishes as at
least the (finite) number of states or the (finite) size of the
acceptance sets strictly decreases. \qed
\end{proof}

\begin{algorithm}
\caption{normal\_form ($S$: \mass): \mass}
\label{algo:fn}
\begin{algorithmic}[1]
  \STATE $S'$ $\leftarrow$ $S$ 
  \REPEAT 
    \STATE unchanged $\leftarrow$ $true$ 
    \FORALL {$q' \in Q'$}
       \IF{$\neg$ attracted($q'$) $\vee$ $\neg \Acc$-consistent($q'$)}  
           \STATE unchanged $\leftarrow$ $false$ \label{alg:c1:start}
           \FORALL {$q \in Q'$ such that $\delta'(q,a) = q'$}
              \STATE{$\delta'(q,a) =$ undefined}
           \ENDFOR
           \STATE{$Q'$ $\leftarrow$ $Q' \setminus \{ q' \}$}\label{alg:c1:end}
       \ENDIF
       \IF{$\neg F,\Acc$-consistent($q'$)}  
           \STATE unchanged $\leftarrow$ $false$ 
           \STATE{$\Acc'(q')$ $\leftarrow$ $\Acc'(q') \setminus \{ \emptyset \}$} \label{alg:c2}
       \ENDIF
       \IF{$\neg \delta,\Acc$-consistent($q'$)}  
           \STATE unchanged $\leftarrow$ $false$ \label{alg:c3:start}
           \FORALL{$a \in \ready'(q') \setminus \bigcup \Acc'(q')$}
              \STATE{$\delta'(q',a) =$ undefined}
           \ENDFOR
           \STATE $\Acc'(q')$ $\leftarrow$ $\{ X \in \Acc'(q') \mid \forall a \in X: \ \delta'(q',a) \mbox{ is defined } \}$ \label{alg:c3:end}
       \ENDIF
       \IF{$Q' = \emptyset$}  
           \RETURN{$S_\bot$}
       \ENDIF
     \ENDFOR
  \UNTIL{unchanged}
  \RETURN $S'$
\end{algorithmic}
\end{algorithm}

As a result of \autoref{prop:nf}, from now on and without loss of
generality, we always assume that \mass{} are in normal form.

At this point, the reader may wonder why attractability and the
previous different forms of consistency are not fully part of the
definition of \masss{}. The reason for this is because, in what
follows, we propose composition operators on \masss{} and it is easier
to define these constructions without trying to preserve these
different requirements. Now if the combination of two \masss{} (which
are now implicitly supposed to be in normal form) gives rise to a
specification violating one of the above requirements then a step of
normalization has to be applied on the result in order to have an
iterative process.
\section{Pre-Quotient Operation of MAS}
\label{sec:prequo}

We first define an operation called \emph{pre-quotient}. Given two
\masss{} $S_1$ and $S_2$, it returns a \mass{} $S_1 \prequotient S_2$
such that the product of any of its models with any model of $S_2$,
\emph{if terminating}, will be a model of $S_1$. Another operation,
defined in \autoref{sec:clean}, will then be used in
\autoref{sec:quotient} to remove the \emph{``if terminating''} assumption.

\begin{definition}[Pre-quotient]
  \label{def:prequo}
  The \emph{pre-quotient} of two \masss{} $S_1$ and $S_2$, denoted $S_1
  \prequotient S_2$, is the \mass{} $(Q_1 \times Q_2, (q_1^0, q_2^0), \delta,
  \Acc, F)$ with:
  \begin{compactitem}
    \item $\Acc (q_1, q_2) = \{X \mid (\forall X_2 \in \Acc_2 (q_2): X \cap X_2
      \in \Acc_1 (q_1)) \land X \subseteq \left(\bigcup \Acc_1 (q_1)\right) \cap
      \left(\bigcup \Acc_2 (q_2)\right)\}$;
    \item $\forall a \in \Sigma:$ $\delta ((q_1, q_2), a)$ is defined
      if and only if there exists $X \in \Acc (q_1, q_2)$ such that $a \in X$ and
      then $\delta ((q_1, q_2), a) = (\delta_1 (q_1, a), \delta_2 (q_2, a))$;
    \item $(q_1, q_2) \in F$ if and only if $q_1 \in F_1$ or $q_2
      \not\in F_2$.
  \end{compactitem}
\end{definition}

\begin{theorem}[Correctness]
  \label{thm:prequo:correct}
  Given two \masss{} $S_1$ and $S_2$ and an automaton $M \models S_1
  \prequotient S_2$, for any $M_2 \models S_2$ such that $M \times M_2$ is
  terminating, $M \times M_2 \models S_1$.
\end{theorem}

\begin{proof}
  Let $\pi_{\prequotient}$ and $\pi_2$ be the simulation relations of $M \models
  S_1 \prequotient S_2$ and $M_2 \models S_2$. Let $\pi \subseteq ((R \times
  R_2) \times Q_1)$ be the simulation relation such that $((r, r_2), q_1) \in
  \pi$ if there exists a $q_2$ such that $(r_2, q_2) \in \pi_2$ and $(r, (q_1,
  q_2)) \in \pi_{\prequotient}$. For any $((r, r_2), q_1) \in \pi$:
  \begin{compactitem}
  	\item $\ready (r, r_2) \in \Acc_1 (q_1)$: by definition of the product of
      automata, we have: $\ready(r, r_2) = \ready(r) \cap \ready(r_2)$ and by definition
      of the acceptance set of the pre-quotient, this intersection is in the
      acceptance set of $q_1$.
  	\item $(r, r_2) \in G_1 \times G_2$ implies $q_1 \in F_1$: by definition of
      the pre-quotient, if $r \in G_1$, then $q_1 \in F_1$.
  	\item for any $a$, if $\lambda ((r, r_2), a) = (r', r'_2)$, then $\delta_1
      (q_1, a)$ is defined and moreover, $((r', r'_2), \delta_1 (q_1, a)) \in \pi$:
      $\lambda(r, a) = r'$, so $\delta((q_1, q_2), a)$ is defined and equal to
      some $(q_1', q_2')$ and, by definition of the pre-quotient, $\delta_1
      (q_1, a) = q_1'$ and $\delta_2 (q_2, a) = q_2'$; $(r, (q_1, q_2)) \in
      \pi_\prequotient$ so $(r', (q_1', q_2')) \in \pi_\prequotient$, $(r_2,
      q_2) \in \pi_2$, so $(r_2', q_2') \in \pi_2$, hence $((r', r_2'), q_1')
      \in \pi$. \qed
  \end{compactitem}
\end{proof}

The specification returned by the quotient is also expected to be
complete, ie., to characterize all the possible automata whose product
with a model of $S_2$ is a model of $S_1$. However, such a
specification may become very large as it will, in particular, have to
allow from a state $(q_1,q_2)$ all the transitions which are not
fireable from $q_2$ in $S_2$. As these transitions will always be
removed by the product with models of $S_2$, they serve no real
purpose for the quotient. We propose to return a compact specification
for the quotient, without these transitions which we then call
\emph{unnecessary}.

An automaton~$M$ is said to have no \emph{unnecessary transitions}
regarding a \mass{}~$S$, denoted $M \unncomp S$, if and only if there
exists a simulation relation $\pi \subseteq R \times Q$ such that
$(r^0, q^0) \in \pi$ and for all $(r, q) \in \pi$, $\ready (r)
\subseteq \bigcup \Acc (q)$ and for every $a$ and $r'$ such that
$\lambda (r, a) = r'$, $(r', \delta (q, a)) \in \pi$.

When an automaton~$M$ has unnecessary transitions regarding a
\mass{}~$S$, it is possible to remove these transitions. Let $\rho_u
(M, S)$ be the automaton $M' = (R \times Q, (r^0, q^0), \lambda', G
\times Q)$ with:
$$\lambda' ((r, q), a) = \left\{\begin{array}{l}
      (\lambda (r, a), \delta (q, a)) \mbox{ if } a \in \bigcup \Acc (q)\\
      \mbox{undefined otherwise}
    \end{array}\right.$$
\begin{theorem}
  \label{thm:unnec}
  Given an automaton $M$ and a \mass{} $S$, we have: $\rho_u (M, S)
  \unncomp S$. Moreover, for all $M_S \models S$, the automata $M
  \times M_S$ and $\rho_u (M, S) \times M_S$ are bisimilar.
\end{theorem}

\begin{proof}
  $\rho_u(M, S) \unncomp S$: let $\pi$ be the simulation relation such that for
  any state $(r, q)$ of $\rho_u(M, S)$, $((r, q), q) \in \pi$; by definition of
  $\rho_u$, $\ready(r, q) \subseteq \bigcup \Acc (q)$.

  $M \times M_S$ and $\rho_u(M, S) \times M_S$ are bisimilar: let $\pi$ be the
  simulation relation such that for any state $(r, r_S)$ of $M \times M_S$ and
  $((r, q), r_S)$ of $\rho_u(M, S) \times M_S$, $((r, r_S), ((r, q), r_S))
  \in \pi$. Then, $\ready(((r, q), r_S)) = (\ready(r) \cap \bigcup \Acc(q)) \cap
  \ready(r_S)$. As $r_S$ implements $q$, $\ready(r_S) \subseteq \bigcup
  \Acc(q)$, so $\ready(((r, q), r_S)) = \ready(r) \cap \ready(r_S) = \ready((r,
  r_S))$. \qed
\end{proof}

  We can then prove that our pre-quotient is complete for automata
  without unnecessary transitions. Given an arbitrary automaton, it
  suffices to remove these transitions with $\rho_u$ before checking
  if it is a model of the quotient.

\begin{theorem}
  \label{thm:prequo:complete}
  Given two \masss{} $S_1$ and $S_2$ and an automaton $M$ such that
  $M\nobreak\unncomp\nobreak S_2$ and for all $M_2 \models S_2$ we have $M \times M_2
  \models S_1$, then $M \models S_1 \prequotient S_2$.
\end{theorem}

\begin{proof}
Let $\pi$ be a simulation relation such that $(r^0, (q_1^0, q_2^0)) \in \pi$ and
for any $(r, (q_1, q_2)) \in \pi$, $a$ and $r'$ such that $\lambda(r, a) = r'$,
$(r', \delta ((q_1, q_2), a)) \in \pi$. This definition of $\pi$ is only correct
if for any $(r, (q_1, q_2)) \in \pi$ and $a$ such that $\lambda (r, a)$ is
defined, $\delta ((q_1, q_2), a) = (\delta_1 (q_1, a), \delta_2 (q_2, a))$ is also
defined. As $M \unncomp S_2$, $a \in \bigcup \Acc_2 (q_2)$, so there exists
an $X \in \Acc_2 (q_2)$ such that $a \in X$ and then $\delta_2 (q_2, a)$ is defined
(as $S_2$ is well-formed). As $\delta_2 (q_2, a)$ is defined, there exists an
automaton $M_2 \models S_2$ with a state $r_2$ implementing $q_2$ such that $(r,
r_2)$ is reachable in $M \times M_2$ and $\lambda_2 (r_2, a)$ is defined. Then,
$\lambda ((r, r_2), a)$ is defined and, as $M \times M_2 \models S_1$, it
implies that $\delta ((q_1, q_2), a)$ is defined.

There are then three points to prove for any $(r, (q_1, q_2)) \in \pi$:
\begin{compactitem}
  \item $\ready (r) \in \Acc ((q_1, q_2))$: by definition of the pre-quotient,
    $\ready (r)$ must verify two properties:
	\begin{compactitem}
		\item $\forall X_2 \in \Acc_2 (q_2): \ready (r) \cap X_2 \in \Acc_1 (q_1)$:

      Let $X_2$ be an element of $\Acc_2 (q_2)$. There exists an automaton $M_2$
      with a state $r_2$ such that $(r, r_2)$ is reachable in $M \times M_2$ and
      $\ready (r_2) = X_2$. Then, as $M \times M_2 \models S_1$ by a simulation
      relation $\pi_\times$ and $((r, r_2), q_1) \in \pi_\times$, $\ready (r)
      \cap \ready (r_2) = \ready (r) \cap X_2 \in \Acc_1 (q_1)$.
    \item $\ready (r) \subseteq \bigcup \Acc_1 (q_1) \cap \bigcup \Acc_2 (q_2)$:

      By definition of $\unncomp$, $\ready(r) \subseteq \bigcup \Acc_2 (q_2)$.

      Assume that $\ready(r) \not\subseteq \bigcup \Acc_1(q_1)$: there is an $a
      \in \ready(r)$ such that $a \not\in \bigcup \Acc_1(q_1)$. As $M$ has no
      unnecessary transition regarding $S_2$, there is a model $M_2$ of $S_2$
      with a state $r_2$ such that $(r, r_2)$ is reachable in $M \times M_2$ and
      $a \in \ready(r_2)$. Then, the transition $((r, r_2), a)$ is defined in $M
      \times M_2$. As $M \times M_2 \models S_1$, the transition $(q_1, a)$ has
      to be defined, which is in contradiction with the hypothesis that $a
      \not\in \bigcup \Acc_1 (q_1)$. Thus, $\ready(r) \subseteq \bigcup \Acc_1
      (q_1)$.
	\end{compactitem}
  \item $r \in G$ implies $(q_1, q_2) \in F_{\prequotient}$, that is $q_1 \in
    F_1$ or $q_2 \not\in F_2$:

    This property is only false if $r \in G$, $q_1 \not\in F_1$ and $q_2 \in
    F_2$. In this case, there exists an automaton $M_2 \models S_2$ with a state
    $r_2$ such that $(r, r_2)$ is reachable in $M \times M_2$ and $r_2 \in G_2$.
    Then, $M \times M_2 \models S_1$ by a simulation relation $\pi_\times$,
    $((r, r_2), q_1) \in \pi_\times$ and $(r, r_2)$ is marked. By definition of
    satisfaction, it implies that $q_1 \in F_1$, which is impossible as we
    already know that $q_1 \not\in F_1$. So $r \in G$ implies $(q_1, q_2) \in
    F_{\prequotient}$.
	\item for any $a$ and $r'$ such that $\lambda (r, a) = r'$, $(r', \delta
    ((q_1, q_2), a)) \in \pi$ is trivial by definition of $\pi$. \qed
\end{compactitem}
\end{proof}

\begin{corollary}[Completeness]
  \label{cor:prequo:complete}
  Given two \masss{} $S_1$ and $S_2$ and an automaton $M$ such that
  for all $M_2 \models S_2$ we have $M \times M_2 \models S_1$, then
  $\rho_u (M, S_2) \models S_1 \prequotient S_2$.
\end{corollary}

\begin{proof}
  By \autoref{thm:unnec}, we know that $\rho_u (M, S_2) \unncomp
  S_2$ and for any $M_2 \models S_2$, $\rho_u (M, S_2) \times M_2$ is bisimilar
  to $M \times M_2$, which implies that $\rho_u (M, S_2) \times M_2 \models
  S_1$. Then, \autoref{thm:prequo:complete} implies that $\rho_u (M, S_2)
  \models S_1 \prequotient S_2$. \qed
\end{proof}

Observe now that the pre-quotient $S_1 \prequotient S_2$ may admit some
models whose product with some models of $S_2$ may not be
terminating. Consider indeed the specifications $S_1$ and $S_2$ of
Fig.~\ref{fig:prequo:S1} and~\ref{fig:prequo:S2} and their pre-quotient
in \autoref{fig:prequo}. The product of the models $M_1^1$ of $S_1
\prequotient S_2$ (\autoref{fig:prequo:M11}) and $M'$ of $S_2$
(\autoref{fig:mod:M11}) has a livelock and thus is not
terminating. One may think that there is an error in the pre-quotient
computation and that it should not allow to realize only $\{a,c\}$,
without $b$ in $\Acc(0,0')$. Indeed, it would forbid the model $M_1^1$,
but it would also disallow some valid models such as $M_1^2$
(\autoref{fig:prequo:M12}), which realizes $\{a,c\}$ in a state and
$\{a,b\}$ in another, thus synchronizing on $b$ with any model of
$S_2$ and allowing the joint reachability of the marked states.

\begin{figure}
\centering
\subfigure[$S_2$]{
  \begin{tikzpicture}
    \node[state] (q0) {$0'$};
    \node[below=0pt of q0,align=center] {$\bigl\{\{a,b\},$\\$\{a,b,c\}\bigr\}$};
    \node[marked,right of=q0] (q2) {$1'$};
    \node[below=0pt of q2] {$\{\emptyset\}$};
    \draw[trans] ($ (q0) + (-.5, 0) $) to (q0);
    \draw[trans,min distance=20pt] (q0) to node[above] {$a$} (q0);
    \draw[trans,bend left] (q0) to node[above] {$b$} (q2);
    \draw[trans,bend right] (q0) to node[below] {$c$} (q2);
  \end{tikzpicture}
  \label{fig:prequo:S2}
}
\subfigure[$S_1 \prequotient S_2$]{
  \begin{tikzpicture}[state/.style={circle,draw,inner sep=1pt,node distance=50pt}]
    \node[state] (q0) {$0,0'$};
    \node[below=0pt of q0,align=center] {$\bigl\{\{a\},\{a,b\},$\\$\{a,c\}\bigr\}$};
    \node[marked,right of=q0] (q2) {$1,1'$};
    \node[below=0pt of q2] {$\{\emptyset\}$};
    \draw[trans] ($ (q0) + (-.6, 0) $) to (q0);
    \draw[trans,min distance=30pt] (q0) to node[above] {$a$} (q0);
    \draw[trans,bend left] (q0) to node[above] {$b$} (q2);
    \draw[trans,bend right] (q0) to node[above] {$c$} (q2);
  \end{tikzpicture}
  \label{fig:prequo}
}
\subfigure[$M_1^1 \models S_1 \prequotient S_2$]{
  \begin{tikzpicture}
    \node[state] (q0) {\textcolor{white}{$0$}};
    \node[marked,right of=q0,node distance=40pt] (q2) {\textcolor{white}{$2$}};
    \draw[trans] ($ (q0) + (-.5, 0) $) to (q0);
    \draw[trans,min distance=20pt] (q0) to node[above] {$a$} (q0);
    \draw[trans] (q0) to node[above] {$c$} (q2);
    \node[below=20pt of q0] {$\ $};
  \end{tikzpicture}
  \label{fig:prequo:M11}
}
\subfigure[$M_1^2 \models S_1 \prequotient S_2$]{
  \begin{tikzpicture}
    \node[state] (q0) {\textcolor{white}{$0$}};
    \node[state,right of=q0] (q1) {\textcolor{white}{$1$}};
    \node[marked] (q4) at ($ (q0)!.5!(q1) + (0, -1.5) $) {\textcolor{white}{$2$}};
    \draw[trans] ($ (q0) + (-.5, 0) $) to (q0);
    \draw[trans,bend left] (q0) to node[above] {$a$} (q1);
    \draw[trans,bend right] (q0) to node[below left] {$c$} (q4);
    \draw[trans,bend left] (q1) to node[below] {$a$} (q0);
    \draw[trans, bend left] (q1) to node[below right] {$b$} (q4);
  \end{tikzpicture}
  \label{fig:prequo:M12}
}
\caption{Example of pre-quotient}
\label{fig:prequotient}
\end{figure}
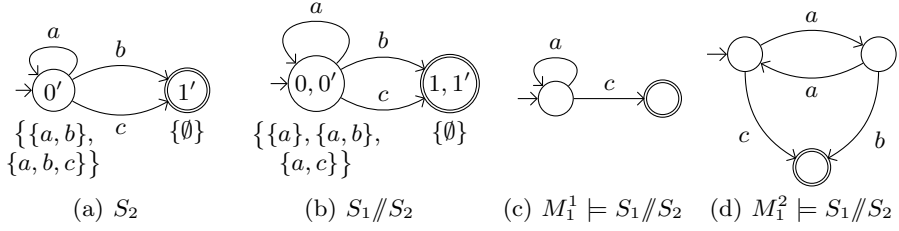

In the next section, we define a criterion allowing to test whether
the product of any models of two \masss{} is terminating or not. On
this basis, we will then refine the pre-quotient in
\autoref{sec:quotient} in order to guarantee the reachability
property.
\section{Compatible Reachability of MAS}
\label{sec:compreach}

By definition, the product of some models of two \masss{} may not
terminate due to two different causes, namely deadlock and livelock. We
consider separately the two issues to derive a compatible reachability
criterion on MAS.

\subsection{Deadlock-free specifications}\label{sec:deadlock}

In this section, we propose a test to check if two \masss{} $S_1$ and
$S_2$ have some models $M_1$ and $M_2$ such that $M_1 \times M_2$ has
a deadlock. To do so, we characterize deadlock-free pairs of states,
from which no deadlock may arise in the product of any two models of
$S_1$ and $S_2$.

Given two acceptance sets $A_1$ and $A_2$, let $\Compat(A_1, A_2)$ be
true iff for all $X_1 \in A_1$ and $X_2 \in A_2$, $X_1 \cap X_2 \neq
\emptyset$. Now a pair of states $(q_1, q_2)$ is said to be deadlock-free,
denoted $\Dead(q_1, q_2)$, if $\Acc_1(q_1) = \Acc_2(q_2) = \{\emptyset\}$ or
$\Compat(\Acc_1(q_1), \Acc_2(q_2))$.

\begin{definition}[Deadlock-free MAS]
  \label{def:deadlock-free}
  Two \masss{} $S_1$ and $S_2$ are \emph{deadlock-free} when all the reachable
  pairs of states in $\Un(S_1) \times \Un(S_2)$ are deadlock-free.
\end{definition}

\begin{theorem}
  \label{thm:deadlock-free}
  Two \masss{} $S_1$ and $S_2$ are deadlock-free if and only if for any
  $M_1\nobreak\models\nobreak S_1$ and $M_2 \models S_2$, $M_1 \times M_2$ is
  deadlock-free.
\end{theorem}

\begin{proof}
  $(\Rightarrow)$ Suppose that $(r_1, r_2)$ is a deadlock in $M_1 \times M_2$.
  Then $(r_1,r_2)$ is not marked and $\ready((r_1,r_2)) = \emptyset$. Now
  $\ready((r_1,r_2)) = \ready(r_1) \cap \ready(r_2)$ and moreover, $(r_1, q_1)
  \in \pi_1$ and $(r_2, q_2) \in \pi_2$ implies $\ready(r_1) \in \Acc_1(q_1)$ and
  $\ready(r_2) \in \Acc_2(q_2)$. As a result, for $X_1 = \ready(r_1) \in \Acc_1
  (q_1)$, $X_2 = \ready(r_2) \in \Acc_2(q_2)$, we have: $X_1 \cap X_2 =
  \emptyset$ and thus $\neg \Compat (\Acc_1(q_1), \Acc_2(q_2))$. Moreover,
  $(r_1,r_2)$ is not marked so $(q_1,q_2)$ is not marked and $\emptyset \not\in
  \Acc_1(q_1)$ and $\emptyset \not\in \Acc_2(q_2)$. In consequence, we have
  $\neg\Dead(q_1,q_2)$ and $S_1$ and $S_2$ are not deadlock-free. \\

  $(\Leftarrow)$ Suppose that $S_1$ and $S_2$ are not deadlock-free: there
  exists $q_1$ and $q_2$ such that $\neg\Dead(q_1,q_2)$. Then there exists $X_1
  \in \Acc_1 (q_1)$ and $X_2 \in \Acc_2 (q_2)$ which verify $X_1 \cap X_2 =
  \emptyset$. For any $M_1 \models S_1$ and $M_2 \models S_2$ with $(r_1, q_1)
  \in \pi_1$ and $(r_2, q_2) \in \pi_2$ such that $\ready(r_1) = X_1$ and
  $\ready(r_2) = X_2$, we have $\ready((r_1,r_2)) = X_1 \cap X_2 = \emptyset$
  in $M_1 \times M_2$. Moreover, $\Acc_1(q_1) \neq \{\emptyset\}$ (or $\Acc_2(q_2)
  \neq \{\emptyset\}$), so there exists a model of $S_1$ (resp. $S_2)$ such that
  a state $r$ implementing $q_1$ (resp. $q_2$) is not marked and has at least
  one transition leading to another marked state, so $(r_1, r_2)$ is not marked.
  As a result, $(r_1,r_2)$ is a deadlock and $M_1 \times M_2$ is not
  deadlock-free. \qed
\end{proof}

\subsection{Livelock-free specifications}\label{sec:livelock}

In this section, we explain how we can check if two \masss{} $S_1$ and
$S_2$ have some models $M_1$ and $M_2$ such that $M_1 \times M_2$ has
a livelock. We identify the cycles shared between $S_1$ and $S_2$
along with the transitions leaving them. We check if at least one of
these transitions is preserved in the product of any two models of
$S_1$ and $S_2$. Before studying these common cycles, a first step
consists in unfolding $S_1$ and $S_2$ so as possible synchronizations
become unambiguous. 

\medskip

\textbf{Unfolding.} \ Given two specifications $S_1$ and $S_2$, we
define the \emph{partners} of a state $q_1$ as $Q_2 (q_1) = \{q_2 \mid
(q_1, q_2)\text{ is reachable in }\Un (S_1) \times \Un (S_2)\}$; the
set $Q_1 (q_2)$ is defined symmetrically. As a shorthand, if we know
that a state $q_1$ has exactly one partner, we will also use
$Q_2(q_1)$ to denote this partner.

If some states of $S_2$ have several partners, it is possible to
transform $S_2$ so that each of its states has at most one partner,
while preserving the set of models of the specification. The unfolding
of $S_2$ in relation to $S_1$ is the specification $((Q_1 \cup
\{q^?\}) \times Q_2, (q_1^0, q_2^0), \delta_u, \Acc_u, (Q_1 \cup
\{q^?\}) \times F_2)$ where:
\begin{compactitem}
  \item $q^?$ is a fresh state ($q_1^?$ denotes a state in $Q_1 \cup
    \{q^?\}$);
  \item $\delta_u ((q_1^?, q_2), a)$ is defined if and only if
    $\delta_2 (q_2, a)$ is defined and then:
    \[\begin{array}{ll}
      \delta_u ((q_1, q_2), a) = &
        \left\{\begin{array}{ll}
          (\delta_1 (q_1, a), \delta_2 (q_2, a)) & \text{if $\delta_1 (q_1, a)$ is defined} \\
          (q^?, \delta_2 (q_2, a)) & \text{otherwise}
        \end{array}\right. \\
      \delta_u ((q^?, q_2), a) = & (q^?, \delta_2 (q_2, a))
    \end{array}\]
  \item $\Acc_u ((q_1^?, q_2)) = \Acc_2 (q_2)$.
\end{compactitem}

Two \masss{} $S_1$ and $S_2$ have \emph{single partners} if and only
if for all $q_1 \in Q_1$, we have $|Q_2(q_1)| \leq 1$ and for all $q_2
\in Q_2$, we also have $|Q_1(q_2)| \leq 1$. 

Given two \masss{} $S_1$ and $S_2$, there exists some \masss{} $S_1'$
and $S_2'$, called \emph{unfoldings} of $S_1$ and $S_2$, with single
partners and which have the same models as $S_1$ and $S_2$. These two
\masss{} can be computed by unfolding $S_1$ in relation to $S_2$ and
then $S_2$ in relation to the unfolding of $S_1$.

\begin{lemma}
  \label{lem:unfolding-models}
  Given two \masss{} $S_1$ and $S_2$ and $S_u$ the unfolding of $S_2$ in
  relation to $S_1$, $S_u \equiv S_2$.
\end{lemma}

\begin{proof}
  $(\Rightarrow)$ Let $M$ be a model of $S_u$. Let $\pi_u$ be the simulation
  relation between the states of $M$ and the states of $S_u$ and let $\pi_2$ be
  the simulation relation such that $(r, q_2) \in \pi_2$ if and only if there
  exists a $q_1^?$ such that $(r, (q_1^?, q_2)) \in \pi_u$. $(r^0, q_2^0) \in
  \pi_2$ and for any $(r, q_2) \in \pi_2$:
  \begin{compactitem}
    \item $\ready(r) \in \Acc_2(q_2)$ as $\ready(r) \in \Acc_u((q_1^?, q_2)) =
      \Acc_2(q_2)$;
    \item if $r \in G$, $q_2 \in F_2$ as $(q_1^?, q_2) \in (Q_1 \cup \{q^?\})
      \times F_2$;
    \item for any $a \in \ready(r)$, $(\lambda (r, a), \delta_2 (q_2, a)) \in
      \pi_2$ as $(\lambda (r, a), \delta_u ((q_1^?, q_2), a)) \in \pi_u$.
  \end{compactitem}
  Thus $M$ is a model of $S_2$. \\

  $(\Leftarrow)$ Let $M$ be a model of $S_2$. Let $\pi_2$ be the simulation
  relation between the states of $M$ and the states of $S_2$ and let $\pi_u$ be
  the simulation relation such that $(r, (q_1^?, q_2)) \in \pi_u$ if and only if
  $(r, q_2) \in \pi_2$ and $(q_1^?, q_2)$ is reachable in $S_u$. $(r^0, (q_1^0,
  q_2^0)) \in \pi_u$ and for any $(r, (q_1^?, q_2)) \in \pi_2$:
  \begin{compactitem}
    \item $\ready(r) \in \Acc_u((q_1^?, q_2))$ as $\ready(r) \in \Acc_2(q_2) =
      \Acc_u((q_1^?, q_2))$;
    \item if $r \in G$, $(q_1^?, q_2) \in (Q_1 \cup \{q^?\}) \times F_2$ as $q_2
      \in F_2$;
    \item for any $a \in \ready(r)$, $(\lambda (r, a), \delta_u ((q_1^?, q_2), a))
      \in \pi_u$ as $(\lambda (r, a), \delta_2 (q_2, a)) \in \pi_2$.
  \end{compactitem}
  Thus $M$ is a model of $S_u$. \qed
\end{proof}

\begin{lemma}
  \label{lem:unfolding-prod}
  Given two \masss{} $S_1$ and $S_2$ and $S_u$ the unfolding of $S_2$ in
  relation to $S_1$, for any $(q_1, (q_1^?, q_2))$ reachable in $\Un (S_1)
  \times \Un (S_2)$, $q_1 = q_1^?$.
\end{lemma}

\begin{proof}
If a state is reachable in $\Un (S_1) \times \Un (S_2)$, there is a path from
the initial state to it. By induction on this path:
\begin{itemize}
  \item if it is empty, we are in the initial state $(q_1^0, (q_1^0, q_2^0))$;
  \item otherwise, we are in a state $(q_1, (q_1, q_2))$ and there is
    a transition by an action $a$ to another state $(\delta_1 (q_1,
    a), \delta_u ((q_1, q_2), a))$. As $\delta_1 (q_1, a)$ is defined,
    $\delta_u ((q_1, q_2), a) = (\delta_1 (q_1, a), \delta_2 (q_2,
    a))$, so the destination state is the pair $(\delta_1 (q_1, a),
    (\delta_1 (q_1, a), \delta_2 (q_2, a)))$. \qed
\end{itemize}
\end{proof}

\begin{lemma}
  \label{lem:unfolding-Qi}
  Given two \masss{} $S_1$ and $S_2$ and $S_u$ the unfolding of $S_2$ in
  relation to $S_1$, for any state $q_u$ of $S_u$, $|Q_1 (q_u)| \leq 1$.
\end{lemma}

\begin{proof}
  Let suppose that $|Q_1 (q_u)| > 1$. Then, there exists at least two
  different states $q_1$ and $q'_1$ such that $(q_1, q_u)$ and $(q'_1,
  q_u)$ are reachable in $\Un (S_1) \times \Un (S_2)$. By the
  definition of the unfolding operation, there exists some $q_1^? \in
  Q_1 \cup \{q_?\}$ and $q_2 \in Q_2$ such that $q_u = (q_1^?, q_2)$.
  By \autoref{lem:unfolding-prod}, $q_1 = q_1^?$ and $q'_1 = q_1^?$,
  so $q_1 = q'_1$. But we know by hypothesis that they are different,
  so $|Q_1 (q_u)| \leq 1$. \qed
\end{proof}

Two \masss{} $S_1$ and $S_2$ have \emph{single partners} if and only
if for all $q_1 \in Q_1$, we have $|Q_2(q_1)| \leq 1$ and for all $q_2
\in Q_2$, we also have $|Q_1(q_2)| \leq 1$. 

\begin{theorem}
  \label{thm:unfoldings}
  Given two \masss{} $S_1$ and $S_2$, there exists some \masss{}
  $S_1'$ and $S_2'$, called \emph{unfoldings} of $S_1$ and $S_2$, with
  single partners and which are equivalent to $S_1$ and $S_2$.
\end{theorem}

\begin{proof}
  Let $S_1'$ be the unfolding of $S_1$ in relation to $S_2$ and $S_2'$ the
  unfolding of $S_2$ in relation to $S_1'$.

  By \autoref{lem:unfolding-models}, we know that $S_1'$ has the same models as
  $S_1$ and $S_2'$ as $S_2$.

  By \autoref{lem:unfolding-Qi}, we know that for any $q_1'$ in $S_1'$,
  $|Q_2(q_1')| \leq 1$ and that for any $q_2'$ in $S_2'$, $|Q_1'(q_2')| \leq 1$.
  Remains to prove that $|Q_2'(q_1')| \leq 1$.

  Let $q_1'$ be a state of $S_1'$. If $|Q_2(q_1')| = 0$, then $|Q_2'(q_1')| = 0$
  as $S_2$ and $S_2'$ have the same models. Otherwise, there exists a $q_2$ such
  that $Q_2(q_1') = \{q_2\}$. There exists then $n$ states (with $n > 0$)
  $q_{2_i}'$ of the form $(q_{1_i}^?, q_2)$. But each $q_{2_i}'$ is in
  relation with at most one state ($q_{1_i}'$) of $S_1'$, as $|Q_1'(q_{2_i}')|
  \leq 1$, and all these $q_{1_i}'$ are different (as the $q_{2_i}'$ are
  different). So there is at most one $q_{2_i}'$ in relation with $q_1'$ and
  thus $|Q_2'(q_1')| \leq 1$. \qed
\end{proof}

\medskip

\textbf{Cycles.} \ In order to detect livelocks, we need to study the
cycles that may be present in the models of a
specification. Intuitively, a cycle is characterized by its states and
the transitions between them.

Given a \mass{} $S$, the partial map $C: Q \rightharpoonup 2^\Sigma$
represents a \emph{cycle} in S if and only if for any $q \in
\dom(C)$,
\begin{inparaenum}[(a)]
  \item $C(q) \neq \emptyset$,
  \item $\exists X \in \Acc(q)$ such that $C(q) \subseteq X$,
  \item $\dom(C) \subseteq \post(q)$ and
  \item $\forall a \in C(q): \delta(q, a) \in \dom(C)$.
\end{inparaenum}

A model $M$ of a MAS $S$ \emph{implements} a cycle $C$ if and only
if there exists a set $\mathcal R$ of states of $M$ such that each $q \in
\dom(C)$ is implemented by at least one state of $\mathcal R$ and for each $r
\in \mathcal R$ and for each $q$ such that $(r, q) \in \pi$,
\begin{inparaenum}[(a)]
  \item $q \in \dom(C)$,
  \item $C(q) \subseteq \ready(r)$,
  \item $\forall a \in C(q): \lambda(r, a) \in \mathcal R$ and
  \item $\forall a \in \ready(r) \backslash C(q): \lambda(r, a) \not\in
    \mathcal R$.
\end{inparaenum}
A cycle is said to be \emph{implementable} if there exists a model $M$ of $S$
implementing the cycle.

We define in \autoref{alg:impl-cycles-rec} an operation, $\CycleI$-rec, which
computes the cycles of a \mass{} passing by a given state. 

\begin{theorem}
  Given a \mass{} $S$, a model $M$ of $S$ implements a cycle $\mathcal C$ if and
  only if $\mathcal C \in \CycleI(S)$.
\end{theorem}

\begin{proof}
  $(\Rightarrow)$ Let $\mathcal C$ be a cycle in $S$ and $M$ a model of $S$
  implementing $\mathcal C$, with $\mathcal R$ the set of states of $M$
  implementing the states of $\mathcal C$. Let $r$ be an element of $\mathcal
  R$ and $q$ a state it implements. By definition, $\CycleI(S)$ contains the
  result of $\CycleI(S, q)$, which calls $\CycleI\text{-rec}(S, q, \emptyset)$.
  For an iteration of the loop at line \ref{alg:impl-cycles-rec:for2}, the
  variable $C$ will take the value of $\mathcal C(q)$ and it will be inserted in
  the generated cycle. The algorithm will then be called recursively on the
  successors of $q$ in the cycle, until $q$ is reached again, thus obtaining
  $\mathcal C$. \\

  $(\Leftarrow)$ Let $\mathcal C$ be a cycle returned by
  $\CycleI(S)$. It is possible to build an automaton $M$ implementing
  the states and transitions of $\mathcal C$. The problem is to make
  sure that this automaton is terminating, ie. that it is possible to
  reach a marked state from any implementation of a state of
  $\dom(\mathcal C)$. By definition of the cycle implementation
  relation, we know that $\dom(\mathcal C) \cap F \neq \emptyset$ or
  $\exists q \in \dom(\mathcal C): \exists X \in \Acc(q): \mathcal C
  (q) \subset X$. In the first case, there is a marked state in the
  loop, thus $M$ is terminating. In the second case, we know that
  there is a state $q$, implemented in $M$ by a state $r$, from which
  there is a transition by an action $a$ which leaves the cycle, that
  is, $\lambda(r, a)$ is not in the set of states implementing
  $\mathcal C$. There is thus no constraint on the transitions from
  $\lambda(r, a)$ and it will be possible to reach a marked state from
  it (provided that $S$ is well-formed).  So $M$ is terminating and in
  consequence is a model of $S$. \qed
\end{proof}

However, some of these cycles may not be implementable. For instance,
the cycle $\mathcal C = \{0 \mapsto \{a\}\}$ is not implementable in
the \mass{} depicted in \autoref{fig:unimpl-cycle}, as all the models
have to eventually realize the transition by $b$ to reach the marked
state and then are not allowed to simultaneously realize the
transition by $a$.  

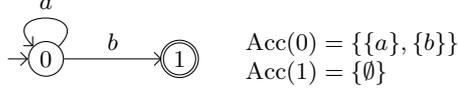
\begin{figure}
  \centering
  \begin{tikzpicture}
    \node[state] (q0) {$0$};
    \node[marked,right of=q0] (q1) {$1$};
    \draw[trans] ($ (q0) + (-.5, 0) $) to (q0);
    \draw[trans,min distance=20pt] (q0) to node[above] {$a$} (q0);
    \draw[trans] (q0) to node[above] {$b$} (q1);
    \node[right of=q1,anchor=west,node distance=20pt] {$\begin{array}{l}
        \Acc(0) = \{\{a\},\{b\}\}\\
        \Acc(1) = \{\emptyset\}
      \end{array}$};
  \end{tikzpicture}
  \caption{A \mass{} over \{a,b\} with no implementable cycle}
  \label{fig:unimpl-cycle}
\end{figure}

In order to be implementable, a cycle has to contain a marked state or it must
be possible to realize a transition that is not part of the cycle in addition to
the transitions of the cycle. Thus, the set of implementable cycles of a
\mass~$S$, denoted $\CycleI(S)$, is $\bigcup_{q \in Q} \{\mathcal C \in
\CycleI\text{-rec}(S, q, \emptyset) \mid \dom(\mathcal C) \cap F \neq \emptyset
\lor \exists q_\mathcal C \in \dom(\mathcal C): \exists X \in \Acc(q_\mathcal
C): \mathcal C (q_\mathcal C) \subset X\}$.

\begin{algorithm}
  \caption{$\CycleI$-rec ($S$: \mass, $q$: State, cycle: Cycle):
    Set Cycle}
\label{alg:impl-cycles-rec}
\begin{algorithmic}[1]
  \INLIF{$q \in \dom(\text{cycle})$}{\RET $\{$cycle$\}$}
  \STATE res $\leftarrow$ $\emptyset$
  \FORALL{$A \in \Acc(q)$}
    \STATE cycle\_acc $\leftarrow$ $\{a \mid a \in A \land q \in \post (\delta (q, a))\}$
      \label{alg:CycleI-rec:cycle-acc}
    \FORALL{$C \in 2^{\text{cycle\_acc}} \backslash \{\emptyset\}$} \label{alg:CycleI-rec:for}
      \label{alg:impl-cycles-rec:for2}
      \STATE current $\leftarrow$ $\{$cycle$\}$
      \FORALL{$a \in C$} \label{alg:CycleI-rec:for2}
        \STATE current $\leftarrow$ $\bigcup_{cycle \in \text{current}} \CycleI\text{-rec} (S, \delta (q, a), cycle \cup \{q \mapsto C\})$
          \label{alg:CycleI-rec:rec}
      \ENDFOR
      \STATE res $\leftarrow$ res $\cup$ current
    \ENDFOR
  \ENDFOR
  \RETURN res
\end{algorithmic}
\end{algorithm}

\textbf{Livelock-freeness.} \ We can now analyze the cycles of two
\masss{} with single partners in order to detect if there may be a
livelock in the product of some of their models. To do so, we
distinguish two kinds of transitions: those, denoted~$\mathcal A$,
which are always realized when the cycle is implemented and those,
denoted~$\mathcal O$, which may (or may not) be realized when the
cycle is implemented. These sets are represented by partial functions
from a state to a set of sets of actions and given, for a particular
cycle $\mathcal C$, by the following formulae:
\begin{gather*}
\mathcal A = \{q \mapsto \text{leaving}(q, A) \mid A \not\in \Acc(q),
(q, A) \in \mathcal C\} \\
\mathcal O = \{q \mapsto \text{leaving}(q, A) \mid A \in\Acc (q), (q,
A) \in \mathcal C \land \text{leaving}(q, A) \neq \emptyset\} \\
\text{where leaving}(q, A) = \{X \setminus A \mid X \in \Acc(q) \land A \subset
X\}
\end{gather*}

\begin{definition}
  \label{def:livelock-test}
  Given two \masss{} $S_1$ and $S_2$ with single partners and a cycle $\mathcal
  C_1$ in $S_1$ such that all its states have a partner, $\mathcal C_1$ is
  livelock-free in relation to $S_2$, denoted $\Live(\mathcal C_1, S_2)$, if and
  only if, when the cycle $\mathcal C_2 = \{Q_2(q) \mapsto \mathcal C_1(q) \mid
  q \in \dom(\mathcal C_1)\}$ is in $\CycleI(S_2)$:
  \begin{compactenum}
    \item\label{def:livelock-test:2} $\mathcal A_{\mathcal C_1} \neq \emptyset$,
      $\mathcal A_{\mathcal C_2} \neq \emptyset$ and there exists $q_1' \in
      \dom(\mathcal A_{\mathcal C_1})$ such that $Q_2(q_1') \in \dom(\mathcal
      A_{\mathcal C_2})$ and $\Compat(\mathcal A_{\mathcal C_1} (q_1'),
      \mathcal A_{\mathcal C_2} (Q_2(q_1')))$, or
    \item\label{def:livelock-test:3} $\mathcal A_{\mathcal C_1} \neq \emptyset$,
      $\mathcal A_{\mathcal C_2} = \emptyset$, $\dom(\mathcal C_2) \cap F_2 =
      \emptyset$ and $\forall q_2' \in \dom(\mathcal O_{\mathcal C_2}):$
      $Q_1(q_2') \in \dom(\mathcal A_{\mathcal C_1})$ and $\Compat(\mathcal
      A_{\mathcal C_1} (Q_1(q_2')), \mathcal O_{\mathcal C_2} (q_2'))$, or
    \item\label{def:livelock-test:4} $\mathcal A_{\mathcal C_1} = \emptyset$,
      $\mathcal A_{\mathcal C_2} \neq \emptyset$, $\dom(\mathcal C_1) \cap F_1 =
      \emptyset$ and $\forall q_1' \in \dom(\mathcal O_{\mathcal C_1}):$
      $Q_2(q_1') \in \dom(\mathcal A_{\mathcal C_2})$ and $\Compat(\mathcal
      O_{\mathcal C_1} (q_1'), \mathcal A_{\mathcal C_2} (Q_2(q_1')))$.
  \end{compactenum}
\end{definition}

\begin{definition}[Livelock-free specifications]
  \label{def:livelock-free}
  Two \masss{} $S_1$ and $S_2$ with single partners are \emph{livelock-free} if
  all the implementable cycles of $S_1$ are livelock-free in relation to $S_2$.
\end{definition}

This definition only tests the implementable cycles of $S_1$. It is not
necessary to do the symmetrical test (checking that the implementable cycles of
$S_2$ verify $\Live$) because we only compare the cycle of $S_1$ with the same
cycle in $S_2$ and the tests of \autoref{def:livelock-test} are symmetric.

The previous definition offers a necessary and sufficient condition to
identify \masss{} which can have two respective models whose product
has a livelock:

\begin{theorem}
  \label{thm:livelock-free}
  Two \masss{} $S_1$ and $S_2$ with single partners are livelock-free if and
  only if for any $M_1 \models S_1$ and $M_2 \models S_2$, $M_1 \times M_2$ is
  livelock-free.
\end{theorem}

\begin{proof}
  $(\Rightarrow)$ Assume that there exists $M_1 \models S_1$, $M_2 \models S_2$
  such that $M_1 \times M_2$ has a livelock, that is, there exists $(r_1, r_2)$
  such that $\Cycle((r_1, r_2)) \neq \emptyset$, $\Cycle((r_1, r_2)) \cap G =
  \emptyset$ and there is no transition $(r', a, r'')$ such that $r' \in
  \Cycle((r_1, r_2))$ and $r'' \not\in \Cycle((r_1, r_2))$.
  \begin{compactitem}
    \item If there exists a cycle $\mathcal C_1 \in \CycleI(S_1)$ which is
      implemented in $M_1$ by the states of $\Cycle(r_1)$ and $\mathcal C_2 =
      \{Q_2(q) \mapsto \mathcal C_1(q) \mid q \in \dom(\mathcal C_1)\}$ is
      implemented in $M_2$ by the states of $\Cycle(r_2)$:
      \begin{compactitem}
        \item if there is no transition leaving $\Cycle(r_1)$, then $\mathcal
          A_{\mathcal C_1} = \emptyset$ and $\dom(\mathcal C_1) \cap F_1 \neq
          \emptyset$, so the three tests of \autoref{def:livelock-test} fail and
          $S_1$ and $S_2$ are not livelock-free; symmetrically $S_1$ and $S_2$
          are not livelock-free if there is no transition leaving $\Cycle(r_2)$;
        \item if there are transitions leaving $\Cycle(r_1)$ and $\Cycle(r_2)$,
          they are not compatible, ie. they have different actions or different
          source states. If in both models, some of these transitions are
          in $\mathcal A$ (they have to be present whenever the cycle is
          implemented), the test \ref{def:livelock-test:2} of
          \autoref{def:livelock-test} will detect that they are not compatible.
          If there are some transitions in $\mathcal A_{\mathcal C_1}$ but none
          in $\mathcal A_{\mathcal C_2}$, test \ref{def:livelock-test:3} will
          detect that $M_2$ may implement a transition that will not be covered
          by the transitions in $\mathcal A_{\mathcal C_1}$. Test
          \ref{def:livelock-test:4} handles the symmetrical case. Finally, if
          there are transitions neither in $\mathcal A_{\mathcal C_1}$ nor
          $\mathcal A_{\mathcal C_2}$, it is always possible to generate a
          livelock and all three tests fail.
      \end{compactitem}
    \item Otherwise, multiple cycles are implemented simultaneously in the model
      by unfolding them or two slightly different cycles are implemented in
      $M_1$ and $M_2$, and then there will also be a livelock in the models
      which implement only one of the cycles, which brings us back to the first
      case.
  \end{compactitem} 

\medskip

  $(\Leftarrow)$ Assume that $S_1$ and $S_2$ are not livelock-free. Then, there
  exists a cycle $\mathcal C_1$ such that $\neg\Live(\mathcal C_1, S_2)$. Then,
  the three conditions of \autoref{def:livelock-free} are all false.
  \begin{compactitem}
  \item If $\mathcal A_{\mathcal C_1} \neq \emptyset$ and $\mathcal
    A_{\mathcal C_2} \neq \emptyset$, then for any state $q_1' \in
    \dom(\mathcal A_{\mathcal C_1})$ in $S_1$, we have $\neg \Compat(
    \mathcal A_{\mathcal C_1} (q_1'), \mathcal A_{\mathcal C_2}
    (Q_2(q_1')))$. So there exists a model $M_1$ of $S_1$ implementing
    $\mathcal C_1$ and a model $M_2$ of $S_2$ implementing $\mathcal
    C_2$ such that there is no transition leaving the cycle in their
    product, hence there is a livelock in $M_1 \times M_2$.
    \item If $\mathcal A_{\mathcal C_1} \neq \emptyset$ and $\mathcal
      A_{\mathcal C_2} = \emptyset$, there exists a state $q_2' \in
      \dom(\mathcal O_{\mathcal C_2})$ such that we have $\neg
      \Compat(\mathcal A_{\mathcal C_1} (Q_1(q_2')), \mathcal
      O_{\mathcal C_2} (q_2'))$. So for any model $M_1$ of $S_1$
      implementing $\mathcal C_1$, its product with a model $M_2$ of
      $S_2$ implementing $\mathcal C_2$ for which the only transition
      leaving the cycle is from an implementation of $q_2'$ will have
      a livelock.
    \item If $\mathcal A_{\mathcal C_1} = \emptyset$ and $\mathcal A_{\mathcal
      C_2} \neq \emptyset$, we are in the case symmetric to the previous one.
    \item If $\mathcal A_{\mathcal C_1} = \emptyset$ and $\mathcal A_{\mathcal
      C_2} = \emptyset$, either one of the specifications has no transitions
      leaving the cycle ($\mathcal O_{\mathcal C_i} = \emptyset$ too), so there
      are some models such that their product has a livelock, or both $\mathcal
      O_{\mathcal C_1}$ and $\mathcal O_{\mathcal C_2}$ are not empty, and then
      there exists an $M_1 \models S_1$ implementing $\mathcal C_1$ such that
      the only transition(s) leaving the cycle is (are) from a state $r_1$ and
      an $M_2 \models S_2$ implementing $\mathcal C_2$ such that the only
      transition(s) leaving the cycle is (are) from a state $r_2$ which is never
      paired with $r_1$ in $M_1 \times M_2$, hence there is a livelock in $M_1
      \times M_2$. \qed
  \end{compactitem}
\end{proof}

\textbf{Specifications with compatible reachability.} \  By combining
the tests for deadlock-free and livelock-free specifications, we can
define a criterion checking if two \masss{} $S_1$ and $S_2$ have some
models $M_1$ and $M_2$ such that $M_1 \times M_2$ is not terminating.

\begin{definition}[Compatible reachability]\label{def:compat-reach}
  Two \masss{} $S_1$ and $S_2$ have a \emph{compatible reachability},
  denoted $S_1~\compreach~S_2$, if and only if they are deadlock-free
  and their unfoldings are livelock-free. They have an incompatible
  reachability otherwise.
\end{definition}

\begin{theorem}
  \label{thm:compreach}
  Given two \masss{} $S_1$ and $S_2$, $S_1 \compreach S_2$ if and only if for
  any $M_1 \models S_1$ and $M_2 \models S_2$, $M_1 \times M_2$ is terminating.
\end{theorem}

\begin{proof}
  By definition, $S_1 \compreach S_2$ if and only if $S_1$ and $S_2$
  are deadlock-free and livelock-free. By Th.~\ref{thm:deadlock-free},
  \ref{thm:unfoldings} and \ref{thm:livelock-free}, this is equivalent
  to: for any $M_1 \models S_1$ and $M_2 \models S_2$, $M_1 \times
  M_2$ is deadlock-free and livelock-free, that is, $M_1 \times M_2$
  is terminating. \qed
\end{proof}

This theorem allows \emph{independent implementability} of \masss{}:
given two \masss{} with compatible reachability, each specification
may be implemented independently from the other while keeping the
guarantee that the composition of the resulting implementations will
be terminating and thus satisfy by construction a
reachability property.
\section{Correction of MAS with Incompatible Reachability}
\label{sec:clean}

We now define an operation that, given two \masss{} $S_1$ and $S_2$
with incompatible reachability, returns a \mass{} refining $S_1$ with
a compatible reachability with~$S_2$.

\subsection{Deadlock correction}

First, given two non-deadlock-free \masss{} $S_1$ and $S_2$, we
propose to refine $S_1$ such that the obtained \mass{} $S'_1$ is
deadlock-free with $S_2$.  For this, we iterate through all the
non-deadlock-free pairs of states $(q_1, q_2)$ and remove the elements
of the acceptance set of $q_1$ which may cause a deadlock, as
described in \autoref{alg:dead-correct}. Note that it may return an
empty specification, because of $\rho$, which then means that for any
model $M_1$ of $S_1$, there exists a model $M_2$ of $S_2$ such that
$M_1 \times M_2$ has a deadlock.

\begin{theorem}[Deadlock correction]
  \label{lem:dead-correct}
  Given two \masss{} $S_1$ and $S_2$, $M_1 \models S_1$ is such that for any
  $M_2 \models S_2$, $M_1 \times M_2$ is deadlock-free if and only if $M_1
  \models \text{dead\_correction}(S_1, S_2)$.
\end{theorem}

\begin{proof}
  $(\Rightarrow)$ Assume that for any $M_1 \models S_1$ and $M_2 \models S_2$,
  $M_1 \times M_2$ is deadlock-free. By \autoref{thm:deadlock-free}, $S_1$ and
  $S_2$ are deadlock-free, which implies that there is no pair of states $(q_1,
  q_2)$ such that $\neg\Dead(q_1, q_2)$. Thus, the set dead\_pairs in
  \autoref{alg:dead-correct} is empty and $\text{dead\_correction}(S_1, S_2) =
  S_1$, so $M_1 \models \text{dead\_correction}(S_1, S_2)$. \\

  $(\Leftarrow)$ Assume that there exists an $M_2 \models S_2$ such that $M_1
  \times M_2$ has a deadlock pair of states $(r_1, r_2)$. By
  \autoref{thm:deadlock-free}, this implies that $S_1$ and $S_2$ are not
  deadlock-free and thus that there exists a pair of states $(q_1, q_2)$
  (implemented by $(r_1, r_2)$) reachable in $\Un(S_1) \times \Un(S_2)$ such
  that $\neg\Dead(q_1, q_2)$. Then, in dead\_correction($S_1$, $S_2$), either
  the acceptance set of $q_1$ has been reduced so that $\Compat(\Acc_1'(q_1),
  \Acc_2(q_2))$ is true and $\Dead(q_1, q_2)$ or $q_1$ is not reachable anymore
  and then $(q_1, q_2)$ is not reachable in $\Un(S_1) \times
  \Un(\text{dead\_correction}(S_1, S_2)$. Consequently, either $\ready(r_1)
  \not\in \Acc_1'(q_1)$ or $(r_1, q_1) \not\in \pi$, and thus $M_1$ is not a model
  of $\text{dead\_correction}(S_1, S_2)$. \qed
\end{proof}

\begin{algorithm}
\caption{dead\_correction ($S_1$: MAS, $S_2$: MAS): MAS}
\label{alg:dead-correct}
\begin{algorithmic}[1]
  \STATE $S'_1$ $\leftarrow$ $S_1$ 
  \FORALL{$(q_1, q_2)$ such that $\neg\Dead(q_1, q_2)$}
    \IF{$\Acc_2(q_2) = \{\emptyset\}$}
      \INLIFELSE{$\emptyset \in \Acc'_1(q_1)$}{$\Acc_1'(q_1) \leftarrow \{\emptyset\}$}{$\Acc_1'(q_1) \leftarrow \emptyset$}
    \ELSE
      \STATE $\Acc_1'(q_1) \leftarrow \{ X_1 \mid X_1 \in \Acc'_1(q_1) \land \forall X_2 \in \Acc_2(q_2): X_1 \cap X_2 \neq \emptyset \}$
    \ENDIF
  \ENDFOR
  \RETURN $\rho(S'_1)$
\end{algorithmic}
\end{algorithm}

\subsection{Livelock correction}

Secondly, given $S_1$ and $S_2$ two deadlock-free \masss{}, we propose to refine
$S_1$ such that the obtained specification $S'_1$ is livelock-free with $S_2$.

There are two ways to prevent livelocks from occuring in the product of the
models of two \masss{}: removing some transitions so that states from which it
is not possible to guarantee termination will not be reached and forcing some
transitions to be eventually realized in order to guarantee that it will be
possible to leave cycles without marked states. For this last method, we
introduce marked acceptance specifications with priorities that are \masss{}
with some priority transitions which have to be eventually realized.

\begin{definition}[\mass{} with priorities]
  A marked acceptance specification with priorities (\massp{}) is a
  tuple $(Q, q^0, \delta, \Acc, P, F)$ where $(Q, q^0, \delta, \Acc,
  F)$ is a \mass{} and $P : 2^{2^{Q \times \Sigma}}$ is a set of
  \emph{priorities}.
\end{definition}

\begin{definition}[Satisfaction]
  An automaton $M$ implements a \massp{} $S$ if $M$ implements the underlying
  \mass{} and for all $\mathcal P \in P$, either $\forall (q, a) \in \mathcal P:
  \forall r: (r, q) \not\in \pi$ or $\exists (q, a) \in \mathcal P: \exists r:
  (r, q) \in \pi \land a \in \ready(r)$.
\end{definition}

Intuitively, $P$ represents a conjunction of disjunctions of
transitions: at least one transition from each element of $P$ must be
implemented by the models of the specification.

Let $S_1$ and $S_2$ be two \masss{} and $q_1$ a state of $S_1$ such
that $q_1$ belongs to a livelock. Then, there exists a cycle $\mathcal
C_1$ in $S_1$ and its partner $\mathcal C_2$ in $S_2$ such that the
conditions given in \autoref{def:livelock-test} are false. Given these
cycles, \autoref{alg:live-correct-cycle} ensures that the possible
livelock will not happen, either by adding some priorities or by
removing some transitions. 

\begin{algorithm}[h]
\caption{live\_corr\_cycle ($S_1$: MASp, $\mathcal C_1$: Cycle, $S_2$: MAS, $\mathcal C_2$: Cycle): MASp}
\label{alg:live-correct-cycle}
\begin{algorithmic}[1]
  \IF{$\mathcal A_{\mathcal C_2} \neq \emptyset$}
    \label{alg:live-correct:1}
    \STATE $Q_A \leftarrow \{q_1 \mid Q_2(q_1) \in \dom(\mathcal A_{\mathcal
      C_2}) \land \forall A \in \mathcal A_{\mathcal C_2}(Q_2(q_1)): A \cap
      \ready(q_1) \neq \emptyset\}$
    \IF{$Q_A \neq \emptyset$}
      \label{alg:live-correct:prio-start}
      \STATE $P \leftarrow \{\bigcup_{1 \leq i \leq |Q_A|} \{(q_i, a)
        \mid a \in X_i\} \mid X_i \in \{A \cap \ready(q_i) \mid A \in \mathcal
        A_{\mathcal C_2}(Q_2(q_i))\}\}$
      \RETURN $(Q_1, q_1^0, \delta_1, \Acc_1, P_1 \cup P, F_1)$
      \label{alg:live-correct:prio-end}
    \ENDIF \label{alg:live-correct:1end}
  \ELSIF{$\dom(\mathcal C_2) \cap F_2 = \emptyset$}
    \label{alg:live-correct:2}
    \STATE $\Acc' \leftarrow \Acc_1$
    \FORALL{$q_1 \in \{Q_1(q_2) \mid q_2 \in \dom(\mathcal O_{\mathcal C_2})\}$}
      \STATE $\Acc'(q_1) \leftarrow \{X \mid X \in \Acc_1(q_1) \land \forall O \in
        \mathcal O_{\mathcal C_2}(Q_2(q_1)): X \cap O \neq \emptyset\}$
    \ENDFOR
    \RETURN $\rho((Q_1, q_1^0, \delta_1, \Acc', P_1, F_1))$
      \label{alg:live-correct:2end}
  \ENDIF\label{alg:live-correct:3}
  \STATE $\Acc' \leftarrow \Acc_1$
  \FORALL{$q_1 \in Q_1$}
    \STATE $\Acc'(q_1) \leftarrow \{X \mid X \in \Acc_1(q_1) \land \forall a \in
      X: \delta(q_1, a) \not\in \dom(\mathcal C_1)\}$
  \ENDFOR
  \RETURN $\rho((Q_1, q_1^0, \delta_1, \Acc', P_1, F_1))$
    \label{alg:live-correct:3end}
\end{algorithmic}
\end{algorithm}

We then iterate over the possible cycles, fixing those which may cause a
livelock , as described in \autoref{alg:live-correct}.

\begin{algorithm}[h]
\caption{live\_correction ($S_1$: MAS, $S_2$: MAS): MASp}
\label{alg:live-correct}
\begin{algorithmic}[1]
  \STATE $S_1' \leftarrow (Q_1, q_1^0, \delta_1, \Acc_1, \emptyset, F_1)$
  \FORALL{$\mathcal C_1 \in \CycleI(S_1)$ s.t. $\forall q_1 \in \dom(\mathcal C_1): |Q_2(q_1)| = 1$ and $\neg\Live(\mathcal C_1, S_2)$}
      \label{alg:live-correct:test-live}
      \STATE $S_1' \leftarrow$ live\_corr\_cycle$(S_1', \mathcal C_1, S_2, \{Q_2(q) \mapsto \mathcal C_1(q) \mid q \in \dom(\mathcal C_1)\})$
  \ENDFOR
  \RETURN $S_1'$
\end{algorithmic}
\end{algorithm}

\begin{theorem}[Livelock correction]
  \label{lem:live-correct}
  Given two \masss{} $S_1$ and $S_2$, $M_1 \models S_1$ is such that for any
  $M_2 \models S_2$, $M_1 \times M_2$ is livelock-free if and only if $M_1
  \models \text{live\_correction}(S_1, S_2)$.
\end{theorem}

\begin{proof}
  $(\Rightarrow)$ Assume that for any $M_1 \models S_1$ and $M_2
  \models S_2$, $M_1 \times M_2$ is livelock-free. By
  \autoref{thm:livelock-free}, $S_1$ and $S_2$ are livelock-free which
  means, by \autoref{def:livelock-free}, that for any implementable
  cycle $\mathcal C_1$ in $S_1$ such that its states have a partner in
  $S_2$, we have $\Live(\mathcal C_1, S_2)$. In this case, the test at
  line \ref{alg:live-correct:test-live} of \autoref{alg:live-correct}
  is always false and so $\text{live\_correction}(S_1, S_2)$ returns
  $S_1$, of
  which $M_1$ is a model by hypothesis. \\

  $(\Leftarrow)$ Assume that there exists an $M_2 \models S_2$ such that $M_1
  \times M_2$ has a livelock.
  \begin{compactitem}
    \item If there exists a cycle $\mathcal C_1 \in \CycleI(S_1)$ which is
      implemented in $M_1$ by the states of the loop in which there is a
      livelock when combined with $M_2$. Thus, live\_correction\_cycle will be
      called with $\mathcal C_1$. There are three cases:
      \begin{compactitem}
        \item If $\mathcal A_{\mathcal C_2}$ is not empty, some transitions are
          present in all the models of $S_2$ implementing $\mathcal C_2$, so the
          models of $S_1$ should realize (at least) one of these transitions
          once. If it is possible, some priorities are added, see lines
          \ref{alg:live-correct:prio-start} to \ref{alg:live-correct:prio-end}
          of \autoref{alg:live-correct-cycle}. This addition will only remove
          the models of $S_1$ that never realize any transition in $\mathcal
          A_{\mathcal C_2}$ and thus that will have a livelock with some models
          of $M_2$ (which only realize the transition of $\mathcal A_{\mathcal
          C_2}$.
        \item If $\mathcal A_{\mathcal C_2}$ is empty but there is no marked
          state in $\mathcal C_2$, all the models of $S_2$ implementing
          $\mathcal C_2$ will eventually realize a transition of $\mathcal
          O_{\mathcal C_2}$ in order to reach a marked state (as there is none
          in the cycle). The only way to avoid a livelock with any model of
          $S_2$ is to realize all the transitions that these models may use to
          reach a marked state, which is done in lines \ref{alg:live-correct:2}
          to \ref{alg:live-correct:2end}.
        \item Otherwise, there will always be a possible livelock with some
          models of $S_2$, so the only possibility is to disallow all the models
          which implement this cycle, which is done in lines
          \ref{alg:live-correct:3} to \ref{alg:live-correct:3end}.
      \end{compactitem}
      So $M_1$ is not a model of the \massp{} returned by
      live\_correction\_cycle for $\mathcal C_1$ and thus it is not a model of
      $\text{live\_correction}(S_1, S_2)$.
    \item Otherwise, multiple cycles are implemented simultaneously and there
      will also be livelocks in the models which implement only one of the
      cycles. As argued in the previous item, applying live\_correction\_cycle
      for these cycles will generate a specification forbidding the
      corresponding models, and then $M_2$ will not be a model of the resulting
      specification as it only combines the behavior of these models. \qed
  \end{compactitem}
\end{proof}

As a result, by applying successively dead\_correction and
live\_correction, we can define the following operation $\rho_\mathcal T$:
\[
\rho_\mathcal T(S_1,S_2) = \text{live\_correction}(\text{dead\_correction}(S_1, S_2), S_2)
\]
Given two \masss{} $S_1$ and $S_2$, it refines the set of models of
$S_1$ as precisely as possible so that their product with any model of
$S_2$ is terminating.

\begin{theorem}[Incompatible reachability correction]
  \label{thm:rho-T}
  Given two \masss{} $S_1$ and $S_2$, $M \models \rho_\mathcal T(S_1, S_2)$ if
  and only if $M \models S_1$ and for any $M_2 \models S_2$, $M \times M_2$ is
  terminating.
\end{theorem}

\begin{proof}
  For any $M \models \rho_\mathcal T(S_1, S2)$ and $M_2 \models S_2$, $M \times
  M_2$ is terminating if and only if $M \times M_2$ is deadlock-free and
  livelock-free. By theorems \ref{thm:deadlock-free} and
  \ref{thm:livelock-free}, this is true if and only if $\rho_\mathcal T(S_1,
  S_2)$ and $S_2$ are deadlock-free and livelock-free, which is true by
  definition of $\rho_\mathcal T$ and Theorems \ref{lem:dead-correct} and
  \ref{lem:live-correct}. \qed
\end{proof}

\section{Quotient Operation of MAS}
\label{sec:quotient}

We can now combine the pre-quotient and cleaning operations to define the
quotient of two \masss{}.

\begin{definition}
  Given two \masss{} $S_1$ and $S_2$, their \emph{quotient} $S_1 /
  S_2$ is given by $\rho_\mathcal T(S_1 \prequotient S_2, S_2)$.
\end{definition}

\begin{theorem}[Soundness]
  Given two \masss{} $S_1$ and $S_2$ and an automaton $M \models S_1 / S_2$, for
  any $M_2 \models S_2$, $M \times M_2 \models S_1$.
\end{theorem}

\begin{proof}
  By \autoref{thm:rho-T}, we know that for any $M_2 \models S_2$, $M \times M_2$
  is terminating. Thus, \autoref{thm:prequo:correct} implies that $M \times M_2
  \models S_1$. \qed
\end{proof}

\begin{theorem}[Completeness]
  Given two \masss{} $S_1$ and $S_2$ and an automaton $M$ such that $\forall M_2
  \models S_2: M \times M_2 \models S_1$, then $\rho_u (M, S_2) \models S_1 /
  S_2$.
\end{theorem}

\begin{proof}
  We know by \autoref{cor:prequo:complete} that $\rho_u (M, S_2) \models S_1
  \prequotient S_2$. We then deduce by \autoref{thm:rho-T} that $\rho_u (M, S_2)
  \models S_1 / S_2$. \qed
\end{proof}

These theorems indicate that each specification $S_2$ and $S_1 / S_2$
may be implemented independently from the other and that the composition of
the resulting implementations will eventually be terminating and will
also satisfy $S_1$.

\section{Conclusion}

In this paper, we have introduced marked acceptance specifications. We
have developed several compositionality results ensuring a
reachability property by construction and, in particular, a sound and
complete quotient. Note that this framework can almost immediately be
enriched with a refinement relation, parallel product and conjunction
by exploiting the constructions available in~\cite{Raclet2008} and
\cite{CaillaudR12}, hence providing a complete specification theory as
advocated in~\cite{Raclet2011}.

Considering an acceptance setting instead of a modal one offers a gain
in terms of expressivity as MAS provide more flexibility than the
marked extension of modal specifications~\cite{CaillaudR12}. 
This benefit becomes essential for the quotient as may/must modalities
are not rich enough to allow for a complete operation. Consider indeed
the two MMS $S_1$ and $S_2$ in \autoref{fig:q-modal} in which optional
transitions are represented with dashed arrows while required
transitions are plain arrows. A correct and complete modal quotient in
this example would tell in the initial state of $S_1 / S_2$ that at
least one action between $a$ and $b$ is required. This cannot be
encoded by modalities but it can be correctly stated by the acceptance
set $\{ \{ a \}, \{ b \}, \{ a,b \}, \{ a, c \}, \{ b, c \}, \{ a, b,
c \}\}$.  

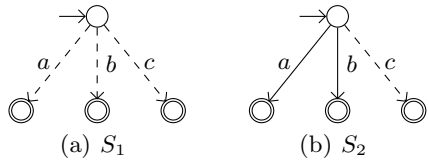
\begin{figure}
\centering
\subfigure[$S_1$]{
\begin{tikzpicture}
\path (1,0) node[state] (0) {$\ $};
\path (0,-1.25) node[marked] (1) {$\ $};
\path (1,-1.25) node[marked] (3) {$\ $};
\path (2,-1.25) node[marked] (2) {$\ $};
\draw[trans] ($ (0) + (-.5, 0) $) to (0);
\draw[trans,dashed] (0) to node[midway, left]{$a$} (1);
\draw[trans,dashed] (0) to node[midway, right]{$b$} (3);
\path (0) edge[trans,dashed] node[midway, right]{$c$} (2);
\end{tikzpicture}}
\hspace{0.5cm}
\subfigure[$S_2$]{
\begin{tikzpicture}
\path (1,0) node[state] (0) {$\ $};
\path (0,-1.25) node[marked] (1) {$\ $};
\path (1,-1.25) node[marked] (3) {$\ $};
\path (2,-1.25) node[marked] (2) {$\ $};
\draw[trans] ($ (0) + (-.5, 0) $) to (0);
\draw[trans] (0) to node[midway, left]{$a$} (1);
\draw[trans] (0) to node[midway, right]{$b$} (3);
\path (0) edge[trans,dashed] node[midway, right]{$c$} (2);
\end{tikzpicture}}
\caption{Two MMS showing that a modal quotient cannot
  exist\label{fig:q-modal}}
\end{figure}

Observe also that quotient of two MAS is heterogeneous in the sense
that its result may be a \massp{}. By definition, \massp{} explicitly
require to \emph{eventually} realize some transitions fixed in the
priority set $P$. By bounding the delay before the implementation of
the transitions, a \massp{} could become a standard \mass{} and the
quotient would then become homogeneous. Algorithms for bounding
\massp{} are left for future investigations.

\bibliographystyle{splncs03}
\bibliography{main}

\end{document}